\documentclass{llncs}
\usepackage{llncsdoc}
\usepackage{graphicx}
\usepackage{amsfonts}
\usepackage{amssymb}
\usepackage{amsmath}
\usepackage{url}

\usepackage{soul}
\usepackage[usenames,dvipsnames]{color}
\usepackage{xcolor}

\begin{document}
\title{Referral-Embedded Provision Point Mechanisms for Crowdfunding of Public Projects}

\author{Praphul Chandra\inst{1} \and Sujit Gujar\inst{2} \and Y. Narahari \inst{3}}

\institute{praphulcs@gmail.com, Indian Institute of Science, Bangalore \and 
sujit.gujar@iiit.ac.in, International Institute of Information Technology, Hyderabad \and hari@csa.iisc.ernet.in, Indian Institute of Science, Bangalore
}

\maketitle
\begin{abstract}
Civic Crowdfunding is emerging as a popular means to mobilize funding from citizens for public projects. A popular mechanism deployed on civic crowdfunding platforms is a provision point mechanism, wherein, the total contributions must reach a predetermined threshold in order for the project to be provisioned (undertaken). Such a mechanism may have multiple equilibria; unfortunately, in many of these equilibria, the project is not funded even if it is highly valued among the agents. Recent work has proposed mechanisms with refund bonuses where the project gets funded in equilibrium if its net value is higher than a threshold \emph{among the agents who are aware} of the crowdfunding effort. In this paper, we formalize the notion of social desirability of a public project and propose mechanisms which use the idea of {\em referrals\/} to expand the pool of participants and achieve an equilibrium in which the project gets funded if its net value exceeds a threshold \emph{among the entire agent population}. We call this new class of mechanisms {\em Referral-Embedded Provision Point Mechanisms\/} (REPPM). We specifically propose two variants of REPPM and both these mechanisms have the remarkable property that, at equilibrium, referral bonuses are \emph{offered} but there is no need for actual payment of these bonuses. We establish that any given agent's equilibrium strategy is to refer other agents and to contribute in proportion to the agent's true value for the project. By referring others to contribute, an agent can, in fact, reduce his equilibrium contribution. In our view, the proposed mechanisms will lead to an increase in the number of projects that are funded on civic crowdfunding platforms.
\end{abstract}

\markboth{Praphul Chandra, Sujit Gujar and Y. Narahari}{Referral-Embedded Provision Point Mechanisms for Civic Crowdfunding}


\section{Introduction}
Civic crowdfunding platforms like Spacehive \cite{spacehive}, Citizinvestor \cite{citizinvestor} and Neighbourly \cite{neighbourly} etc., aim to generate funding for public and community projects from citizens. In the United Kingdom alone, Spacehive  has generated \textsterling $5$ million for over $150$ public projects from citizen contributions across $68$ cities. A typical process that is followed in crowdfunding of public projects is as follows:

\begin{figure}
\centering{\includegraphics[scale=0.4]{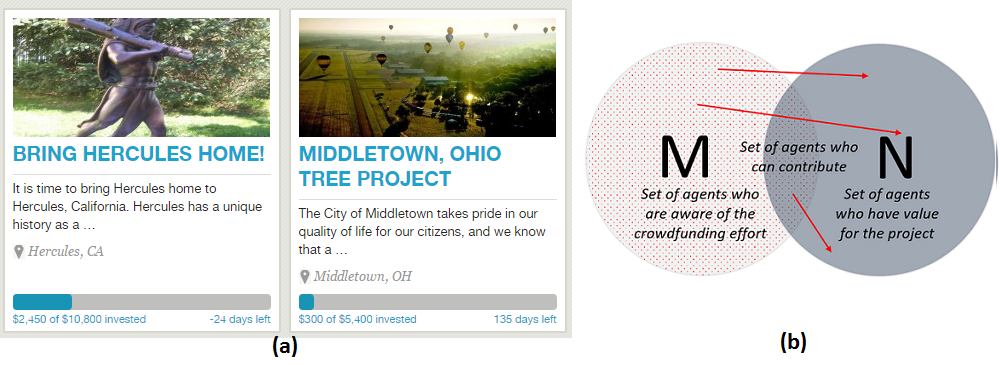}}
\caption{(a) Projects on a civic crowdfunding platform (b) Referrals can expand the set of agents who can contribute.}
\label{Motivation}
\end{figure}

\begin{enumerate}
\item\emph{Requester posts a public project proposal}: A requester, seeking crowdfunding for a public project, posts a proposal. The proposal specifies a target amount of funds to be raised for the project to be provisioned:  the target amount is known as the \emph{provision point}. The requester also specifies a deadline by which the funds need to be raised.
\item\emph{Agents arrive}: Agents arrive over time to view the project and observe (a) the target amount,  (b) the amount pending to be funded, and (c) the deadline. 
\item\emph{Agents contribute}: An agent may contribute any amount to the project.
\item\emph{Requester provisions or refunds}: If the funding target is achieved by the deadline, the requester provisions the project; otherwise, the contributions of all agents are refunded.
\end{enumerate}
\noindent
We refer to this as the \emph{Provision Point Baseline} (PPB) mechanism. The class of \emph{Provision Point Mechanisms} (PPM) we consider share the following characteristics: (i) {if the provision point is reached, the contributions are collected and the project is provisioned} and (ii) {if provision point is not reached, the project is not provisioned and the contributions are refunded; in addition, a bonus may be paid to agents}. In PPB, no bonus is offered or paid. In the \emph {Provision Point Mechanism with Refund Bonus} (PPR) \cite{zubrickas2014} and \emph{Provision Point Mechanism with Securities} (PPS) \cite{chandra2016}, bonuses are offered to incentivize contributions. In this  work, we propose the class of \emph{Referral-Embedded Provision Point Mechanisms} (REPPM) and instantiate two mechanisms of this class: \emph {Referral-Embedded Provision Point Mechanism with Refund Bonus} (REPP-R)  and \emph {Referral-Embedded Provision Point Mechanism with Securities} (REPP-S).

\section{Problem Statement, Related Work, Our Contributions}
\subsubsection{Problem Statement:} Civic crowdfunding has had mixed results with only 44\% of posted projects getting funded\cite{spacehive}. An important question is to understand why a project does not get funded. There are at least three distinct challenges:
\begin{enumerate}
\item {\textsc{Challenge-1:} the project is not valued enough in the agent population it is purported to benefit.}
\item {\textsc{Challenge-2:} the project is valued enough by the agents but not all the agents who value it were aware of the crowdfunding effort.}
\item {\textsc{Challenge-3:} the project is valued enough by the agents, all the agents were aware of the crowdfunding effort, but some agents chose to free ride on the contributions of others.} 
\end{enumerate}
\textsc{Challenge-3} may be attributed to the use of PPB mechanism which has been shown to have multiple equilibria, in many of which the project is not funded \cite{bagnoli1989,brubaker1975free,schmidtz1991limits}. To solve \textsc{Challenge-3}, Zubrickas \emph{et.al} \cite{zubrickas2014} and Chandra \emph{et.al.} \cite{chandra2016} propose PPR and PPS respectively. Our current work is motivated by \textsc{Challenge-2} on crowdfunding platforms, where, a subtle yet critical distinction must be made between the set of agents who \emph{value} the public project ($N$) and the set of agents who are \emph{aware} of the crowdfunding effort ($M$). The set of agents who can contribute is $M \cap N$ (Figure \ref{Motivation}(b)). To solve \textsc{Challenge-2}, one approach is to use a referral mechanism to incentivize agents who are aware of the crowdfunding effort to refer others who might value the project. With agents in a social network, referral mechanisms can expand the pool of participants, thus increasing the funds created through crowdfunding. This alone however may not be sufficient; we  need  mechanisms which solve \textsc{Challenge-2} and \textsc{Challenge-3} together. This is the problem we address in this paper.

\subsubsection {Challenges in Designing Referrals for Provision Point Mechanisms:}
Our first challenge in introducing a referral mechanism in provision point mechanisms is to ensure that agents do not free-ride and do, in fact, contribute to the public project. Note that, on its own, a referral mechanism \emph{increases} the incentive to free ride since an agent may rely on referral bonuses for gaining his utility. Our second challenge is that, since the project is \emph{public} (non-excludable, non-rival), it is apriori unclear who would pay the referral bonus. In fact, no agent may be willing to pay a referral bonus: this is a key difference with other referral mechanisms in the literature where there exists a center (henceforth, sponsor) who benefits from the referrals. Third, it is desirable that an agent's referral bonus is proportional to the contributions of his referrals and his equilibrium contribution is proportional to his true value for the project (fairness).

\subsubsection {Related Work in Crowdfunding:} Our work is most closely related to the work of Zubrickas \emph{et.al} \cite{zubrickas2014} and Chandra \emph{et.al.} \cite{chandra2016} which we discuss in detail in Section \ref{REPPM}. There is significant literature in the design of mechanisms for the private provisioning of pubic projects \cite{bagnoli1989,brubaker1975free,chandra2016,chen1999,groves1977,schmidtz1991limits,tabarrok1998private,zubrickas2014}. In the {\em voluntary contribution mechanism\/} (VCM), agents voluntarily contribute and the extent of the public project provisioned corresponds to the aggregate funds collected. VCM induces a simultaneous move game which has multiple equilibria and in many of these equlibiria, the public project is not provisioned \cite{healy2006}. Morgan et. al. \cite{morgan2000} studies the use of state lotteries to incentivize contributions to public projects, wherein, a higher contribution leads to a higher likelihood of winning: the game induced attains a unique equilibrium which outperforms VCM. Marx et. al. \cite{marx2000} consider a setting where agents make repeated contributions in a round-robin fashion and prove the existence of a Nash equilibrium where an agent contributes if and only if others have contributed their equilibrium contributions. 

\subsubsection {Related Work in Referral Mechanisms:} Referral mechanisms have been used in a wide variety of settings like the red balloon challenge \cite{nath2012mechanism,pickard11}, viral marketing \cite{abbassi11,arthur09,chandra2016b,drucker12,emek2011mechanisms,lobel2015customer}, and query propagation in social networks \cite{dixit2009quality,kleinberg05}. These mechanisms may be classified based on whether the sponsor of the referral bonus seeks to maximize the spread of information in the network (e.g. viral marketing) or find an (a set of) agent(s) to achieve an objective (e.g. red balloon, query propagation). In either case, however, the sponsor pays out the referral bonus contingent on some observable action or target being achieved. 

\subsubsection{Our Contributions:} Our primary contribution is to design a class of mechanisms for civic crowdfunding of public projects which achieve an equilibrium that is \emph{socially desirable}
(informally, there is enough value of the project to all the agents) 
we define social desirability more formally later on (Definition \ref{def:sc})). We call the proposed class of mechanisms {\em Referral-Embedded Provision Point Mechanisms\/} (REPPM) and quantify the advantage of REPPM over PPM in terms of social desirability. We instantiate two variants of REPPM: REPP-R and REPP-S corresponding to PPR \cite{zubrickas2014} and PPS \cite{chandra2016} respectively. In both these mechanisms, agents are incentivized to contribute using a \emph{refund bonus} and are incentivized to refer other agents to contribute using a \emph{referral bonus} - the latter relies on a {\em Referral Bonus Function\/} (RBF). In REPP-R, both the bonuses are computed using a proportional scheme and in REPP-S, they are computed using an underlying prediction market. 

We believe that in the context of civic crowdfunding of public projects, our work is the first  to design a referral mechanism which \emph{offers} a referral bonus that incentivizes referrals but remarkably \emph{does not pay} the referral bonus in equilibrium. Exploiting this, we show that if the agent population which will benefit from the project forms a connected graph via social relations, our mechanisms achieve an equilibrium where the project is funded (Theorem \ref{thm:REPP_PPR_NE}, Theorem \ref{thm:REPP_PPS_NE}). We also show that an agent's equilibrium contribution is proportional to the agent's value for the project, \emph{less a referral bonus} which depends on the net contributions due to the agent's referrals. We note that both PPR and PPS require a \emph{sponsor} who \emph{offers} bonuses to incentivize contributions. A key advantage of our approach is that it reduces the sponsor's risk: this should make it easier to find such sponsors and thus increase the success of civic crowdfunding. 

The rest of the paper is organized as follows. In Section \ref{Notation}, we set up the notation that we use in the rest of the paper. In Section \ref{REPPM}, we formalize the notion of social desirability, introduce the notion of embedding a referral bonus function in provision point mechanisms and specify the conditions that such a function must satisfy to be used REPPM. In Section \ref{REPPR} and Section \ref{REPPS}, we instantiate REPPM corresponding to PPR and PPS and study the impact of doing so on the equilibrium. We conclude in Section \ref{Conclusion} with a summary.

\section{Notation, Setup, and Assumptions} \label{Notation}
\subsubsection{Notation:} Let $N$ be the set of agents who value a given public project and let $M$ be the set of agents who are aware of the crowdfunding effort. Hence the set of agents who can contribute funds to the project is $M\cap N$ (See Figure \ref{Motivation}(b)). The value that agent $i$ derives from the public project getting provisioned is $\theta_i$ and the net value for the project among agents who can contribute is $\vartheta_{M \cap N} = \sum_{i=1}^{|M \cap N|} \theta_i$. Let $h^0$ be the target amount that needs to be collected for the project to be provisioned. Agent-$i$'s contribution is $x_i \in [0, h^0]$ and the net contribution is $\chi = \sum_{i=1}^{|M \cap N|} x_i$. The vector of contributions is $\mathbf{x} = (x_1,...,x_{|M\cap N|}) \in \mathbb{R}_+^{|M \cap N|}$. We use the subscript $-i$ to represent all agents other than agent $i$, for example, $\mathbf{x}_{-i}$ refers to the vector of contributions of all agents except $i$. Agent $i$ may refer $M_i \subseteq N_i$ other agents, where $N_i$ is set of his neighbors in the social network. 

In a sequential setting, at $t=0$, the requester posts a proposal for funding a public project. This includes the target amount of funds $h^0$ (the provision point) and a deadline $T$ till which agents may contribute to the project. $h^t$ refers to the target amount that remains to be collected at time $t$. Agent-$i$ arrives at time $a_i \in [0,T]$ and observes the funds that have been collected so far ($h^0-h^{a_i}$). Agent $i$ may decide to contribute funds $x_i \in [0,h^{a_i}]$ at any time $t_i \in [a_i,T]$. Thus, in the sequential setting, agent $i$'s strategy, $\psi_i$, consists of his contribution ($x_i$), his time of contribution ($t_i$) and the set of agents he refers ($M_i$) and his utility is $u_i(\psi;\theta_i)$. Table \ref{tab2} (Appendix \ref{ssec:notation}) summarizes key notation.

\begin{figure}
\centering{\includegraphics[scale=0.54]{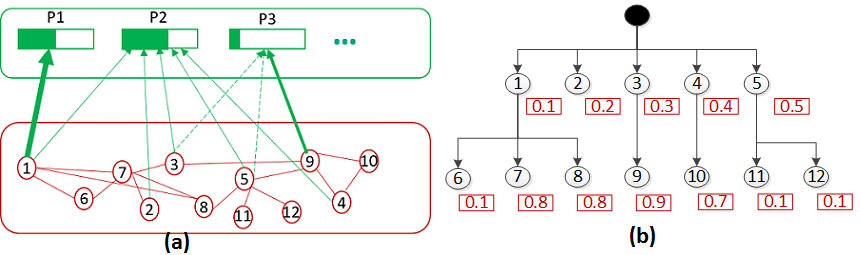}}
\caption{(a) Agent's contributions to projects; (b) Referral tree for P2 with agent types.}
\label{RefTree}
\end{figure}
\noindent
Agents in $M \cap N$ can be represented as a directed graph with the sponsor as the root. If more than one agent refers the same agent, the earliest referral takes precedence. Agents who contribute without being referred by another agent form the sponsor's single hop neighbors. Thus, the referral graph is a tree. Consider, for example, the scenario in Figure \ref{RefTree}(a) where three public projects are requesting funds from 12 agents ($|N| = 12)$. An edge from an agent to a project represents that the agent is aware of the effort (visited the project page). The weight of the edge represents an agent's contribution to the project: we use a dotted edge to represent a contribution of value zero. For $P2$,  $M = \{1,2,3,4,5\}$ are aware of the crowdfunding effort and have contributed; if agents $\{1,3,4,5\}$ refer their neighbors, we get the referral tree of Figure \ref{RefTree}(b). 

We make the following assumptions.  \textsc{Assumption-1}: Agents have quasi-linear utility \cite{bagnoli1989,chandra2016,zubrickas2014}. \textsc{Assumption-2}: Apart from  knowing the history of contributions, agents do not have any information regarding whether the project will get funded or not \cite{chandra2016,zubrickas2014}. \textsc{Assumption-3} : The set of agents who have a non-zero value for the project ($N$) forms a connected graph. {\textsc{Assumption-4:} In a sequential setting, agents contribute only once to the project (agents typically visit the project website once and contribute if the project has value to them). Our mechanisms ensure that agents have no advantage in delaying or splitting up their contributions. \textsc{Assumption-5}: An agent's value for the public project ($\theta_i$) is his private information and $h^0$, $T$, and $h^t$ are common knowledge.

\section {Referral-Embedded Provision Point Mechanisms} \label {REPPM}
In the class of PPM, an agent's utility can be stated as follows: 
\begin{definition}{(Un)Funded Utility:} In the class of provision point mechanisms, the (un)funded utility of agent-$i$ is his utility if the target amount is (not) collected and the public project is (not) provisioned.
\end{definition}
In the provision point mechanisms we consider, an agent's funded utility is always $(\theta_i - x_i)$ but mechanisms differ in the unfunded utility. We let $\mathcal{I}_X$ be an indicator random variable which takes the value $1$ if $X$ is true and $0$ otherwise.

\subsection {Provision Point Baseline (PPB) Mechanism}
In PPB, an agent's strategy space consists only of contribution to be made, hence $\psi_i = x_i \quad \forall i$. His unfunded utility is zero and hence his utility is:
\begin{eqnarray}
u_i(\mathbf{x};\theta_i) &=& \mathcal{I}_{\chi \geq h^0}\times(\theta_i - x_i)  + \mathcal{I}_{\chi < h^0} \times 0
\end{eqnarray}
PPB has been shown to have multiple equilibria, many of which are inefficient \cite{bagnoli1989}: a result which has been verified empirically too \cite{healy2006}.

\subsection {Provision Point Mechanism with Refund Bonus (PPR)}
In PPR \cite{zubrickas2014}, if the funding target is not achieved, the contributions are refunded and an additional refund bonus is paid to agents who volunteered to contribute. The refund bonus is $\frac{x_i}{\chi} B \; \forall i$ where $B > 0$ is the refund budget specified at the beginning and is common knowledge among all agents. An agent's strategy space in PPR consists only of his contribution, hence $\psi_i = x_i \quad \forall i$ and his utility is:
\begin{eqnarray}
u_i(\mathbf{x};\theta_i) &=& \mathcal{I}_{\chi \geq h^0}\times (\theta_i - x_i) + \mathcal{I}_{\chi < h^0}\times \left(\frac{x_i}{\chi}B \right) \label{PPRutility}
\end{eqnarray}
The set of Pure Strategy Nash equilibria (See Appendix for definition) with PPR is characterized as follows: 
\begin{theorem}{\cite{zubrickas2014}}
Let $\vartheta_{M \cap N} > h^0$ and $B > 0$. In PPR, the set of PSNE are $\{(x_i^*) : x_i^* \leq  \frac{h^0}{B+h^0}\theta_i \forall i ; \chi = h^0\}$ if $B \leq \vartheta_{M \cap N} - h^0$. Otherwise the set of PSNE is empty. 
\label{PPR_NE}
\end{theorem}
PPR considers a setup where agents decide their contributions \emph{simultaneously} without knowledge of contributions made by the other agents: it considers a simultaneous move game. 

\subsection{Provision Point Mechanism with Securities (PPS)}
In a sequential setting where agents arrive over time and can observe the contributions collected thus far (e.g. civic crowdfunding platforms), the PPS mechanism is better suited. In PPS \cite{chandra2016}, if the funding target is not achieved by the deadline $T$, the contributions are refunded and an additional refund bonus is paid to agents who volunteered to contribute. The refund bonus is designed so that early contributions are incentivized. PPS uses a complex prediction market \cite{abernethy2013efficient} to determine the refund bonus with the key idea being that contributors actually buy contingent securities ($r_i^{t_i}$) which each pay a unit amount if the project is not funded. An agent's strategy space in PPS consists of the quantum and timing of his contribution, hence $\psi_i = (x_i,t_i) \quad \forall i$. Thus, his utility is given as:
\begin{eqnarray}
u_i(\psi;\theta_i) &=& \mathcal{I}_{\chi \geq h^0}(\theta_i - x_i) + \mathcal{I}_{\chi < h^0}(r_i^{t_i} - x_i )  \label{REPPSUtil1} 
\end{eqnarray}
PPS achieves an equilibrium at which the project is funded and thus the refund bonus is not paid out at equilibrium \cite{chandra2016}.

\subsection{Referral Embedded Provision Point Mechanisms (REPPM)} \label{RBFsection}
Similar to PPM, in REPPM too, the project is provisioned only if the collected funds reach the provision point. If the provision point is not reached, the contributions are refunded and an additional bonus is paid to agents who volunteered to contribute. This bonus consists of two parts (i) a refund bonus and (ii) a referral bonus. The refund bonus is calculated using the underlying provision point mechanism while the referral bonus is calculated using a Referral Bonus Function (RBF). The referral bonus incentivizes agents to refer other agents to contribute to the public project. The key intuition in REPPM is to embed a RBF in provision point mechanisms such that it impacts \emph{only} the unfunded utility of agents: since the unfunded utility is realized only if the project is not funded, the referral bonus is paid out only if the project is not provisioned. Thus, REPPM, follow a two pronged approach:
\begin{enumerate}
\item {Design a referral mechanism where a referral bonus is \emph{offered} but is not paid out if the project is funded.}
\item {Intelligently embed the referral mechanisms in provision point mechanisms so that the public project is funded at equilibrium.}
\end{enumerate}

\subsubsection{Advantage of REPPM}: To quantify the advantage of REPPM, we first formalize the notion of social desirability:
\begin{definition}
\label{def:sc}
($N,\tau$) Socially Desirable : A public project is said to be ($N,\tau$) socially desirable if the net value of the project among agents in the set $N$ is greater than $\tau$, that is, $\vartheta_N = \sum_{i=1}^n \theta_i > \tau$.
\end {definition}
In Figure \ref{RefTree}(b), if the cost of $P2$ were $2$, then the project is not socially desirable without referrals but is socially desirable with referrals. PPR \cite{zubrickas2014} ensures that the project gets funded at equilibrium if it is $(M \cap N, h^0)$ socially desirable. In a sequential setting, PPS \cite{chandra2016} ensures that the project gets funded at equilibrium if it is $(M \cap N, C_0^{-1}(h^0 + C_0(0)))$ socially desirable\footnote{In Section \ref{REPPS}, we will explain the $C_0$ function in more detail.}. Thus, in both PPR and PPS mechanisms,  the social desirability condition is based on $M \cap N$. We design mechanisms that are $(N, \tau)$ socially desirable rather than $(M \cap N, \tau')$ socially desirable. We show that this can be achieved if the RBF satisfies the following:

\begin{enumerate}
\item {\textsc{RBF-Condition-1 (Continuous and Differentiable)}  This condition requires that the gradient of the RBF ($\frac{ds}{dR}$) is well defined everywhere so that the marginal increase in referral bonus due to increase referred contribution is well defined.}
\item {\textsc{RBF-Condition-2 (Monotonically Increasing)} This condition requires that the gradient of the RBF is positive (\emph{$s'(R) = \frac{d s}{d R} > 0 \quad \forall R \in \mathbb{R_+}$}) so that an agent has an incentive to refer all agents in his network.} 
\item {\textsc{RBF-Condition-3 (Bounded Loss)} This condition requires that the referral bonus is upper bounded ($s(R) < \sigma \quad \forall R \in \mathbb{R}_+$) so that the loss of the sponsor is upper bounded. To this end, we will required that \emph{$\frac{d^2 s}{dR^2} \leq 0 \quad \forall R$}}. 
\end{enumerate}
Finally, if an agent does not refer, the agent does not get any referral bonus ($s(0) = 0$) so that in the absence of referrals, the mechanism reduces to the underlying provision point mechanism. Though this is not strictly required, it makes the analysis simpler. Some examples of functions that can be used as RBF are $\tanh(R)$, $\left(\frac{1}{1 + \exp\left(- R\right)} - 0.5\right)$ and $\frac{2}{\pi} \arctan (R)$. The choice among these depends on the minimum bonus that needs to be offered to incentivize referrals. We now discuss two instantiations of REPPM.

\section {Embedding Referrals in PPR : REPP-R} \label{REPPR}
REPP-R embeds referrals in PPR. Similar to PPR, in REPP-R, the project is provisioned only if the collected funds reach the provision point. If the provision point is not reached, the contributions are refunded and an additional bonus is paid to agents who volunteered to contribute. This bonus consists of two parts (i) a refund bonus and (ii) a referral bonus. Let $X_{M_i} = \sum_{j \in M_i} x_j$. In REPP-R, agent-$i$'s strategy is $\psi_i = (x_i, t_i, M_i)$ and he has a utility:
\begin{eqnarray}
u_i(\psi;\theta_i) &=& \mathcal{I}_{\chi \geq h^0}(\theta_i - x_i) + \mathcal{I}_{\chi < h^0} \left(\frac{x_i}{\chi}B + s(X_{M_i}) \right) \label{REPPRutility}
\end{eqnarray}
Comparing Equation \eqref{REPPRutility} with Equation \eqref{PPRutility}, we can observe that the unfunded utility in REPP-R contains an additional term which depends on the contributions of agents referred by agent $i$.

\subsection{Impact of Introducing Referral Bonus in PPR} \label{REPPRImpact}
\begin{figure}
\centering{\includegraphics[scale=0.38]{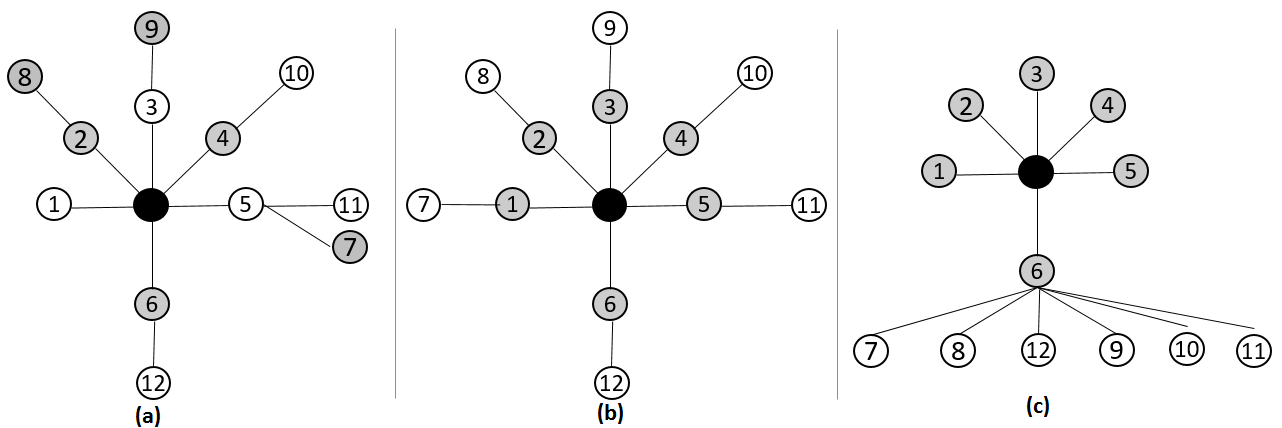}}
\caption{Some Referral Tree Structures in REPP-S}
\label{ReferralBonus}
\end{figure}
\noindent
To understand the impact of introducing referrals in REPP-R, we evaluate the maximum referral bonus that may need to be paid out: this depends on the RBF and the structure of the underlying referral tree. We can show (see Appendix) that the maximum referral bonus needs to be paid out when the provision point ($h^0$) is achieved by $n = |N|$ contributions of the smallest possible contribution $\delta = \frac{h^0}{n}$ and each contributing agent is referred by a chain of $d$ non-contributing agents where $d$ is the diameter of the social network of $N$ agents. The maximum referral bonus paid out is $nd \times s(\delta) < nd\sigma$. Figure \ref{ReferralBonus}(b) shows such a worst case with $d=1$ with the sponsor as the root and the shaded nodes indicating agents who did not contribute. 

\subsection{Equilibrium Analysis of REPP-R}
\begin{theorem}
Let $s()$ be a referral bonus function that satisfies \textsc{RBF-Conditions 1-3}. If REPP-R is used for crowdfunding a project with provision point $h^0$ when $\sigma < \frac{\vartheta_N - h^0 - B}{nd}$, the strategies in the set $\bigg\{(\psi_i^* = \{x_i^*, N_i\}) : x_i^* \leq  \min \left(0,\frac{\theta_i - \sigma}{1+\frac{B}{h^0}} \right)$, otherwise $x_i^* = 0$; $\chi = h^0 )\bigg\}$ are Nash equilibria.
\label{thm:REPP_PPR_NE}
\end{theorem}

\begin{proof}
First, we claim in Step-0 that it is a (weakly) dominant strategy for all agents to refer so that $M \cap N = N$. In Step 1, we show that, at equilibrium, $\chi = h^0$. In Step 2, we characterize the equilibria strategy of agent $i$ ($\psi_i^*$). Step 3 proves the upper bound on $\sigma$. 

\noindent{\underline{Step 0}}: \textsc{RBF-Conditions-1,2,3} ensure that every agent has an incentive to refer since an agent's unfunded utility increases monotonically with his referrals: $s(X_{M_i})$ increases monotonically with $M_i$ and is independent of his contribution. Thus, $M \cap N = N$.

\noindent{\underline{Step 1}}:
If $\chi > h^0$, any agent with a positive contribution can gain in utility by marginally decreasing his contribution.  $\chi < h^0$ cannot hold in equilibrium since, in REPP-R the unfunded utility always increases with contribution ($\frac{x_i}{\chi} B$).\footnote{In Step-3, we show that the upper bound on $\sigma$ ensures that the referral incentives do not override the incentives to contribute.} and $\vartheta_N > (nd\sigma + h^0 + B) > h^0$ means that there exists at least one agent $j \in N$ who can increase his (unfunded) utility by contributing more so that he get a higher refund bonus. 
Thus, in equilibrium $\chi = h^0$. 

\noindent{\underline{Step 2}}: Due to \textsc{Assumption-2}, agents do not have any bias in believing whether the project will be funded, other than the contributions. From Step 1, the contributions would be such that the project is funded in  equilibrium. Thus, at equilibrium, an agent will contribute such that his funded utility is no less than the highest possible unfunded utility, that is $\forall i$:
\begin{eqnarray}
\theta_i - x_i^* \geq  \frac{x_i^*}{h^0}B + s(X_{M_i}^*)  
\quad \text {or equivalently} \quad 
x_i^* \leq \left(\frac{\theta_i - s(X_{M_i}^*)}{1 + \frac{B}{h^0}} \right) \leq \left( \frac{\theta_i - \sigma}{1 + \frac{B}{h^0}} \right) \nonumber 
\end{eqnarray}
where the last inequality follows because even if agents are optimistic about referral bonus and go conservative for $x_i$,  $s(X_{M_i}^*) \leq \sigma$ (\textsc{RBF-Condition-3}). Since negative contributions (withdrawals) are not allowed, a negative equilibrium contribution means that an agent will refer but not contribute. 

\noindent{\underline{Step 3}}:
Summing up $\left(\theta_i - x_i^* \geq  \frac{x_i^*}{h^0}B + s(X_{M_i}^*)\right)$ for all agents and using $\sum_{i \in N} s(X_{M_i}^*) \leq nd\sigma$ (See Section \ref{REPPRImpact}),  we get the condition 
\begin{eqnarray*}
\vartheta_N - h^0 \geq  B + \sum_{i \in N} s(X_{M_i}^*) \quad \Rightarrow \quad 
\sigma < \frac{\vartheta_N - h^0 - B}{nd}
\end{eqnarray*}
\end{proof}
\noindent
\subsubsection{Explanation:} The upper bound\footnote{In theory, no lower bound on referral bonus is needed since any referral incentive, no matter how small, should incentivize referrals. In practice, a lower bound may depend on the effort required to contribute.} on $\sigma$ has a natural interpretation: if the \emph{referral} bonus is higher, it reduces the incentives for an agent to \emph{contribute} to an extent that the project does not get funded at equilibrium. In PPR, the condition for equilibrium is $B < (\vartheta_{M \cap N} - h^0)$, that is, the excess value ($\vartheta_{M \cap N} - h^0$) is used to incentivize contributions. In REPP-R, the excess value has to support the incentives for contribution \emph{and} the incentives for referrals ($B + nd\sigma < \vartheta_N - h^0$ ): with the important difference that the excess value is calculated in a larger pool ($N$ in REPP-R instead of $M \cap N$ in PPR). This means that in scenarios where $\vartheta_N > \vartheta_{M \cap N}$, REPP-R can achieve funding for projects which would not have been funded with PPR as long as the referral bonus is upper bounded appropriately.  

\section{Embedding Referrals in PPS : REPP-S} \label{REPPS}

REPP-S embeds referrals in PPS \cite{chandra2016}. PPS uses a prediction market to determine the refund bonus with the key idea being that contributors are allotted contingent securities \cite{arrow1964role} which each pay a unit amount if the project is not funded. The authors set up a binary prediction market with two outcomes: (i) the project is funded (ii) the project is not funded. PPS allots securities for the project-not-funded outcome to agents who contribute. The number of securities associated with the project-funded outcome is $0 \quad \forall t \in [0,T]$. The number of securities allotted to an agent depends on the contribution and timing of his contribution. To determine the number of securities to allot, PPS leverages a complex prediction market \cite{abernethy2013efficient} created using a cost function $C:\mathbb{R}^2 \rightarrow \mathbb{R}$. To be used in PPS, a cost function must satisfy the following conditions: (i) Path Independence (ii) Continuous and Differentiable (iii) Information Incorporation (iv) No arbitrage (v) Bounded Loss \cite{abernethy2013efficient,chandra2016}. Let $q^t$ denote the total number of securities (associated with project-not-funded outcome) allotted till time $t$ in PPS. The number of securities allotted to agent $i$ if he contributes $x_i$ at time $t_i$ is: 
\begin{eqnarray}
r_i^{t_i} &=& C^{-1}_0(x_i + C_0(q^{t_i}))  - q^{t_i} \label{PPS_Securities}
\end{eqnarray}
where $C_0:\mathbb{R} \rightarrow \mathbb{R}$ is a function derived from $C$ by setting the number of the securities associated with the project-funded outcome  to $0 \quad \forall t \in [0,T]$ \cite{chandra2016}.

In REPP-S, the project is provisioned only if the collected funds reach the provision point. If the provision point is reached, the contributions are collected and neither the refund bonus nor the referral bonus is paid. If the provision point is not reached, the contributions are refunded and an additional bonus is paid to agents who volunteered to contribute. This bonus consists of two parts (i) a refund bonus and (ii) a referral bonus : \emph {both} of these are determined by an underlying prediction market. For agent-$i$, the refund bonus depends only on his contribution and is determined using Equation \eqref{PPS_Securities}. The referral bonus of agent-$i$ depends on the number of securities awarded to the agents referred by him; which, in turn, depends on their quantum and timing of contributions. The total number of securities allocate to agent $i$ in REPP-S is:
\begin{eqnarray}
\rho_i &\triangleq& r_i^{t_i} + s(R_{M_i}) \label{REPPSalloc}
\end{eqnarray}
where $R_{M_i} \triangleq \sum_{j \in M_i}r_j^{t_j}$ is the total number of securities allotted to agents referred by $i$. We make two key observations (i) $R_{M_i}$ depends both on the quanta of contributions generated due to agent-$i$'s referrals and the timing of those contributions: the earlier the referred contributions, the higher the referral bonus (ii) the securities associated with the refund bonus are allocated at $t_i$ as soon as agent-$i$ contributes but the securities associated with the referrals are allocated at $T$ to account for any contributions that may come in due to agent-$i$'s referrals. In REPP-S, agent-$i$'s strategy is $\psi_i = (x_i, t_i, M_i)$ and he has a utility:
\begin{eqnarray}
u_i(\psi;\theta_i) &=& \mathcal{I}_{\chi \geq h^0}(\theta_i - x_i) + \mathcal{I}_{\chi < h^0}(\rho_i - x_i )  \label{REPPSUtil2} 
\end{eqnarray}
\noindent 
Comparing Equation \eqref{REPPSUtil2} with Equation \eqref{REPPSUtil1} shows how REPP-S differs from PPS due to the securities allocated for referred contributions. 

\subsection{Impact of Introducing Referral Bonus in PPS}
To understand the impact of introducing referrals in REPP-S, we evaluate the maximum referral bonus that may need to be paid out: this depends on the total number of securities issued which in turn, depends on the cost function of the prediction market, the referral bonus function and the structure of the underlying referral tree. With an analysis similar to the REPP-R case, we can show (see Appendix) that the total number of securities issued is:
\begin{eqnarray}
q_{max} = \sum_{i=1}^{|M \cap N|} \rho_i = \sum_{i=1}^{|M \cap N|} (r_i^{t_i} + s(R_{M_i})) < C_0^{-1}(h^0 + C_0(0)) + nd\sigma \label{case1bound1}
\end{eqnarray}

\noindent
A higher $q_{max}$ means lower liquidity in the underlying prediction market and hence a lower incentive for contribution. In PPS, to ensure that agents contribute and the project gets funded at equilibrium, an agent's unfunded utility must be a monotonically increasing in his contribution. This, in turn, requires that the cost function must be sufficiently liquid \cite{chandra2016}. In REPP-S, ensuring that an agent's unfunded utility is monotonically increasing in his contribution requires:
\begin{eqnarray*}
\frac{\partial}{\partial x_i} \left (\rho_i - x_i \right) = \frac{\partial}{\partial x_i} \left (r_i^{t_i} + s(R_{M_i}) - x_i \right)  &>& 0 \quad \forall q^{t_i}, \forall x_i < h^0  \nonumber
\end{eqnarray*}
Since $s(X_{M_i})$ is independent of $x_i$ and $r_i^{t_i}$ monotonically decreases with $q^{t_i}$, this condition translates to:
\begin{eqnarray}
\quad \frac{\partial r_i^{t_i}}{\partial x_i}|_{q_{max}} = \frac{\partial}{\partial x_i} (C^{-1}_0(x_i + C_0(q_{max}))  - q_{max}) &>& 1 \label{ePPS}
\end{eqnarray}
Equation \eqref{case1bound1} and Equation \eqref{ePPS} determine the total bonus the sponsor must offer in REPP-S which is higher than in PPS. The advantage of a higher bonus is an increase in the participant pool and thus a higher likelihood of the project getting funded which in turn reduces the sponsor's risk.

\subsection{Equilibrium Analysis of REPP-S} 
\label{REPPSequilibrium}
\begin{theorem}
Let $C_0$ be an appropriate cost function and let $s()$ be a RBF that satisfies \textsc{RBF-Conditions 1-3}. If REPP-S is used for crowdfunding a project with provision point $h^0$ in a social network of $n$ agents with diameter $d$ and if $\sigma < \frac{\vartheta_N - C_0^{-1}(h^0 + C_0(0))}{nd}$, then the strategies in the set $\bigg\{(\psi_i^* = \{x_i^*, a_i, N_i\}) : x_i^* \leq  \min (0, (C_0(\theta_i - \sigma +q^{a_i}) - C_0(q^{a_i})))$ if $h^{a_i} > 0$, otherwise $x_i^* = 0$; $\chi = h^0\bigg\}$ are sub-game perfect  equilibria.
\label{thm:REPP_PPS_NE}
\end{theorem}

\begin{proof} We claim in Step 0 that it is a (weakly) dominant strategy for all agents to refer so $M \cap N = N$. In Step 1, we show that, at equilibrium, $\chi = h^0$. In Step 2, we characterize the equilibria strategy of agent-$i$. Step 3 proves the upper bound on $\sigma$. We show that these equilibria are sub-game perfect in Step 4. 

\noindent{\underline{Step 0}}: This is similar to Step 0 of the proof of Theorem \ref{thm:REPP_PPR_NE}.

\noindent{\underline{Step 1}}: In equilibrium, $\chi > h^0$ cannot hold since the requester stops collecting the funds at $\chi = h^0$. Since, in REPP-S the unfunded utility always increases with contribution (See Equation \eqref{ePPS}) and $\vartheta_N > (nd \sigma + C_0^{-1}(h^0 +C_0(0))) > h^0$, $\chi < h^0$  mean that there is at least one agent who can increase his unfunded utility by contributing more and hence at equilibrium $\chi = h^0$.

\noindent{\underline{Step 2}}: Due to \textsc{Assumption-2}, agents do not have any bias in believing whether the project will be funded, other than the contributions. From Step 1, the contributions would be such that the project is funded in  equilibrium. Thus, at equilibrium, an agent will contribute such that his funded utility is no less than the highest possible unfunded utility. That is, $\rho_i^* - x_i^* \leq \theta_i - x_i^*$ or $\rho_i^*  \leq \theta_i$. Using Equation \eqref{PPS_Securities} and Equation  \eqref{REPPSalloc}, we get
\begin{eqnarray}
C_0^{-1}(x_i^* + C_0(q^{t_i}))  - q^{t_i}  \leq \theta_i - s(R_{M_i}^*) \quad \text{or equivalently} \nonumber \\
x_i^* \leq \bigg( C_0 \left(\theta_i - s(R_{M_i}^*) + q^{t_i}\right) -  C_0(q^{t_i}) \bigg) \leq \bigg(C_0 \left(\theta_i - \sigma + q^{t_i}\right) -  C_0(q^{t_i}) \bigg) \quad   \label{equix} 
\end{eqnarray}
where the last inequality follows from \textsc{RBF-Condition-3} since $s(R_{M_i}) \leq \sigma \quad \forall i$.
Since negative contributions (withdrawals) are not allowed, a negative equilibrium contribution means that an agent will refer but not contribute. 
\noindent Now, note that (i) the  RHS of Equation \eqref{equix} is a monotonically decreasing function of $q^{t_i^*}$ and (ii) $q^{t}$, the number of securities allotted by the market at time $t$, is a monotonically non-decreasing function of $t$. Thus, an agent with value $\theta_i$ minimizes the RHS at $t_i^* = a_i$, that is, he contributes as soon as he arrives. Thus, $t_i^* = a_i$. Intuitively, if an agent delays his contribution, to be indifferent between funded and unfunded utility, the agent needs to contribute more.

\noindent{\underline{Step 3}}: Summing up $\left(\rho_i^* - x_i^*  \leq \theta_i - x_i^*\right)$ for all agents leads to the condition $\sum_{i=1}^n \rho_i^* \leq \vartheta_N$. Using the bound on $\sum_{i=1}^n \rho_i$ from Equation \eqref{case1bound}, we get 
\begin{eqnarray}
\sigma < \frac{\vartheta_N - C_0^{-1}(h^0 + C_0(0))}{nd} \label{boundSigma}
\end{eqnarray}

\noindent{\underline{Step 4}}: These equilibria, specified as a function of the aggregate history ($h^{a_i}$), are also sub-game perfect (See Appendix for definition). Consider agent $j$ who arrives last at $a_j$. If $h^{a_j} = 0$, then his best strategy is $x_j^* =0$. If $h^{a_j} > 0$, irrespective of the history of the contributions and $h^{a_j}$, his funded and unfunded utility are the same at $x_j^*$, defined in the theorem and still it is best response for $j$ to follow the equilibrium strategy. With backward induction, by similar reasoning, it is best response for every agent to follow the equilibrium strategy, irrespective of the history, as long as others follow the equilibrium strategy. Further, no agent has an incentive to delay his contribution either (Assumption-2 and cost of securities never decreases in REPP-S). Thus, these equilibria are sub-game perfect.
\end{proof} 

\noindent {\bf Explanation:} In PPS, the condition for equilibrium is $\vartheta_{M \cap N} < C_0^{-1}(h^0 + C_0(0))$ and the excess value is used to sponsor a prediction market which issues securities for contributions. In REPP-S, the excess value has to support the incentive for contribution \emph{and} the incentives for referrals but the excess value is calculated in a larger pool too ($N$ in REPP-S instead of $M \cap N$ in PPS). The referral incentives can be either carved out of the same budget (lower liquidity in the prediction market) or can be paid out from additional budget (increase in the sponsor's budget). Finally, note that PPS is a \emph{class} of mechanisms. One instantation of PPS is (Logarithmic Market Scoring Rule)LMSR-PPS and our analysis applies to the whole class (See Appendix for LMSR -REPPS). 

\section{Conclusion } \label{Conclusion}
We considered civic crowdfunding, formalized the notion of social desirability of a public project, and proposed {\em Referral-Embedded Provision Point Mechanisms\/} (REPPM), a class of  mechanisms that achieve an equilibrium in which the project gets funded if it is socially desirable \emph{among the whole agent population}. By incentivizing agents to both contribute and refer other agents, REPPM solves two challenges: (i) agents do not free ride (every agent's equilibrium strategy is to contribute in proportion to his true value for the project) and (ii) information about the crowdfunding effort diffuses in the social network so that agents who have value for the project have an opportunity to contribute. This arises at the cost of a higher budget that a sponsor must furnish. However, since neither the referral bonus nor the refund bonus needs to be paid out at equilibrium, finding a sponsor who \emph{offers} these incentives is more likely.  With these advantages, our mechanisms can significantly improve the success rate of civic crowdfunding.
\begin{table}
\centering
{\small
\begin{tabular}{|c|c|c|c|c|} \hline
\textbf{Mechanism}&\textbf{Equilibrium Contribution}&\textbf{Social Desirability}&\textbf{Conditions}\\ \hline
PPR & $\frac{\theta_i}{1 + \frac{B}{h^0}}$ & $(M \cap N, h^0+B)$ & $B \in (0, \vartheta_{M \cap N} - h^0)$  \\ \hline
REPP-R & $\min \left(0,\frac{\theta_i - \sigma}{1 + \frac{B}{h^0}} \right)$ & $(N, h^0+B + nd\sigma))$ & $B \in (0, \vartheta_N - h^0 - nd\sigma)$ \\ \hline
PPS & $C_0(\theta_i+q^{a_i}) - C_0(q^{a_i})$ & $(M \cap N, C^{-1}_0(h^0 + C_0(0)))$ & Sufficient Liquidity \cite{chandra2016} \\ \hline
REPP-S & $\min (0, (C_0(\theta_i - \sigma +q^{a_i}) - C_0(q^{a_i})))$ & $ (N, (C_0^{-1}(h^0 + C_0(0)) + nd\sigma ))$ & Equation \eqref{ePPS} \\ \hline
LMSR-PPS & $b\ln \left(\frac{1 + \exp \left(\frac{\theta_i +  q^{a_i}}{b} \right)}{1 + \exp(\frac{q^{a_i}}{b})} \right)$ & $(M \cap N, h^0 + b \ln 2)$ & $b \in \left (0, \frac{\vartheta_{M \cap N} - h^0}{\ln 2} \right)$ \\ \hline
LMSR-REPP-S & $\min \left(0,b\ln \left(\frac{1 + \exp \left(\frac{\theta_i - \sigma +  q^{a_i}}{b} \right)}{1 + \exp(\frac{q^{a_i}}{b})} \right) \right)$ & $(N, h^0 + b \ln 2 + nd\sigma )$ & $b \in \left (0, \frac{\vartheta_{N} - h^0 - nd\sigma}{\ln 2} \right)$ \\ \hline
\end{tabular}}
\caption{Key Results}
\label{tab1}
\end{table}
\noindent
Comparing REPPM with the corresponding PPM (Table \ref{tab1}) shows that REPPM increases the pool of agents who can contribute, at the cost of increasing the threshold of social desirability. Thus, they are well suited in scenarios where the increase in the participant pool can significantly outweigh the increase in threshold. In web based civic crowdfunding platforms, where the successful funding of a public project requires a significant effort to attract \emph {contributors} and \emph{funds}, our mechanisms will have a significant impact.

\newpage
\bibliography{ecai}

\newpage
\section{Appendix}
\subsection {Notation Table}
\label{ssec:notation}
The following table summarizes the notation used in this paper.
\begin{table}
\centering
\begin{tabular}{|c|p{9.5 cm}|} \hline
\textbf{Symbol}&\textbf{Definition}\\ \hline
$T$ & Time at which fund collection concludes \\ \hline
$t$ & Epoch of time in the interval $[0, T]$ \\ \hline 
$h^t$ & Amount that remains to be funded at $t$; \\ \hline
$h^0$ & Target amount (provision point) \\ \hline
$i \in \{0,1, \ldots \}$ & Agent id; $i=0$ refers to the requester  \\ \hline
$N_i$ & Neighbors of agent $i$ in the social network \\ \hline
$M_i$ & Set of contributors referred to by agent$i$ \\ \hline
$\theta_i \in \mathbb{R}_+$ & Agent $i$'s value for the project \\ \hline
$x_i \in \mathbb{R}_+$ & Agent $i$'s contribution to the project \\ \hline
$a_i \in [0,T]$ & Time at which agent $i$ arrives at the platform \\ \hline
$t_i \in [a_i, T]$ & Time at which agent $i$ contributes to the project\\ \hline
$\psi_i$ & Strategy of agent $i$ \\\hline
$\vartheta_N \in \mathbb{R}_+$ & Net value for the project among agent set $N$\\ \hline
$\chi \in \mathbb{R}_+$ & Net contribution for the project \\ \hline
\end{tabular}
\caption{Key notation}
\label{tab2}
\end{table}

\subsection{Equilibrium Definitions} \label{Problem}
We seek to design mechanisms in a sequential setting such that a public project gets funded at equilibrium. Such mechanisms induce a game among the agents $\{1,2,\ldots,n\}$. With $\psi_i$s being agents' strategies and $u_i$s as their utilities,  we define Pure Strategy Nash Equilibrium (PSNE) and Sub-Game Perfect Equilibrium (SGPE).
\begin{definition}{(Pure Strategy Nash Equilibrium)}
A strategy profile $\psi^* = (\psi_1^*,\ldots,\psi_n^*)$ is said to be a Pure Strategy Nash Equilibrium (PSNE) if $ \forall i, \forall \theta_i $
\begin{eqnarray*}
u_i(\psi_i^*, \psi_{-i}^* ; \theta_i) &\geq& u_i( \tilde{\psi}_i, \psi_{-i}^*;\theta_i) \quad \forall \tilde{\psi}_i.
\end{eqnarray*}
\end{definition}

\noindent
Let $H^t$ be the history of the game till time $t$, that contains the agents' arrivals and their contributions, then we define:
\begin{definition}{(Sub-game Perfect Equilibrium)}
A strategy profile $\psi^* = (\psi_1^*,\ldots,\psi_n^*)$ is said to be a sub-game perfect  equilibrium if $ \forall i, \forall \theta_i$
\begin{eqnarray*}
u_i(\psi_i^*, \psi_{-i\mid H^{a_i}}^* ; \theta_i) &\geq& u_i( \tilde{\psi}_i, \psi_{-i\mid H^{a_i}}^*;\theta_i) \quad \forall \tilde{\psi}_i,  \forall H^t
\end{eqnarray*}
\end{definition}

\noindent
Here $\psi_{-i\mid H^{a_i}}^*$ indicates that the agents who arrive after $a_i$ follow the strategy specified in $\psi_{-i}^*$.

\subsection{REPP-R Worst Case Analysis}
As the number of referrals needed per unit of contribution increases, more referral bonus needs to be paid out. Since the exact amount of referral bonus depends on the referral tree structure, we analyze two possible worst case scenarios. 

\subsubsection {Case-1:}: The provision point ($h^0$) is achieved by $n$ contributions of the smallest possible contribution $\delta = \frac{h^0}{n}$ each  - all of them referred by a different agent. Figure \ref{ReferralBonus}(b) shows such an example where the provision point is met by the contribution of agents in the set $\{7,8,9,10,11,12\}$ each one referred by a different agent. In this example, the contribution is inter-mediated by exactly one referring agent: in general, the path length between a contributor and the sponsor may consist of $d$ unique agents who do not contribute - $d$ being the diameter of the underlying social network. The total referral bonus paid out is $nd \times s(\delta) < nd\sigma$.

\subsubsection {Case-2:}: The provision point ($h^0$) is achieved by $n$ contributions of the smallest possible contribution $\delta = \frac{h^0}{n}$ each - all of them referred by the same agent. Figure \ref{ReferralBonus}(c) shows such an example where the provision point is met by the contribution of agents in the set $\{7,8,9,10,11,12\}$ all of them referred by $6$. In this example, the contribution is inter-mediated by exactly one referring agent: in general, the path length between a contributor and the sponsor may consist of $d$ nodes who monopolize the contributions - $d$ being the diameter of the underlying social network. The total referral bonus paid out in this case is $d \times s(n \delta) < d\sigma$.

\textsc{RBF-Condition-3} ensures that the RBF is a concave function so that the worst case is Case-1.
\subsection{REPP-S Worst Case Analysis}
Since the amount of referral bonus depends on the referral tree structure, we analyze two possible (worst case) scenarios. 

\subsubsection {Case-1:}: The provision point is met by the smallest allowed contributions ($\delta$) and each of these contributions is due to a referral. Thus, the provision point ($h^0$) is achieved by $n$ contributions of $\delta = \frac{h^0}{n}$ each. Figure \ref{ReferralBonus}(b) shows such an example. The total number of securities issued is $\sum_{i=1}^n \rho_i = \sum_{i=1}^n (r_i^{t_i} + s(R_{M_i}))$ which can be expressed in terms of the cost function used in the prediction market and the RBF as:
\begin{eqnarray}
\sum\limits_{i=0}^{n-1} \bigg(C_0^{-1}(\delta + C_0(i \delta)) + d \times s(C_0^{-1}(\delta + C_0(i \delta))) \bigg) &=& 
C_0^{-1}(h^0 + C_0(0)) + d \times \sum\limits_{j=0}^{n-1}  s(C_0^{-1}(\delta + C_0(j\delta)))) \nonumber 
\end{eqnarray}
where the first term follows since the cost function used in the prediction market is path independent \cite{chandra2016}.  Since the RBF and cost function are monotonically increasing, we also have:
\begin{eqnarray*}
d \times \sum\limits_{j=0}^{n-1}  s(C_0^{-1}(\delta + C_0(j\delta)))) \leq nd \times s(C_0^{-1}(\delta + C_0(0)))) 
\end{eqnarray*}
Finally, \textsc{RBF-Condition-3} ensures that $nd \times s(C_0^{-1}(\delta + C_0(0)))) < nd\sigma$, so:
\begin{eqnarray*}
\sum\limits_{i=0}^{n-1} \bigg(C_0^{-1}(\delta + C_0(i \delta)) + d \times s(C_0^{-1}(\delta + C_0(i \delta))) \bigg) < C_0^{-1}(h^0 + C_0(0)) + nd\sigma \label{case1bound}
\end{eqnarray*}

\subsubsection {Case-2:}: The provision point is met by the smallest allowed contributions ($\delta$) and all the contributions are referred by a single agent. Figure \ref{ReferralBonus}(c) shows such an example. The total number of securities issued is:
 \begin{eqnarray}
\sum\limits_{i=0}^{n-1} \bigg (C_0^{-1}(\delta + C_0(i \delta)) \bigg) + d \times s\left(\sum\limits_{i=0}^{n-1} C_0^{-1}(\delta + C_0(i \delta)) \right) 
&\leq& C_0^{-1}(h^0 + C_0(0)) + d\sigma \nonumber \label{case2bound}
 \end{eqnarray}
Which of the two cases is applicable in a given scenario depends on the RBF. Specifically, Case-1 applies when 
\textsc{RBF-Condition-3} ensures that the RBF is a concave function so that $\sum_{i=0}^{n-1} \left(s(C_0^{-1}(\delta + C_0(i \delta))) \right) >  
s\left(\sum_{i=0}^{n-1} C_0^{-1}(\delta + C_0(i \delta)) \right)$ and thus the worst case is Case-1.

\subsection{LMSR - REPP-S} \label{LMSRREPPS}
In a market with a binary outcome event, let the vector of outstanding securities be $(q_{\omega_0},q_{\omega_1})$. The number of securities associated with the project-funded outcome ($q_{\omega_1}$) is $0 \forall t \in [0,T]$. Let $q^t$ denote the total number of securities (associated with project-not-funded outcome) allotted till time $t$ in PPS. A popular cost functions that satisfies the conditions to be used in PPS is the LMSR (Logarithmic Market Scoring Rule) based cost function \cite{hanson2012logarithmic}: 
\begin{eqnarray*}
C(q_{\omega_0}, q_{\omega_1}) = b\ln(\exp(q_{\omega_0}/b) + \exp(q_{\omega_1}/b)) \Rightarrow C_0(q_0^t) = b\ln(1 + \exp(q_0^t/b))
\end{eqnarray*}
where $b$ is a parameter that controls the market liquidity. Applying Equation \eqref{equix} for an LMSR based REPP-S, the equilibrium contribution of agent $i$ with value $\theta_i$ who arrives at $a_i$ and refers $M_i$ agents is :
\begin{eqnarray*}
x_i^* &\leq& b\ln \left(1 + \exp\left(\frac{\theta_i  - \sigma +  q_0^{a_i}}{b}\right)\right) - b\ln \left(1 + \exp(\frac{q_0^{a_i}}{b}) \right) \\
\end{eqnarray*}
Using Equation \eqref{case1bound1}, the maximum number of securities that can be allocated is:
\begin{eqnarray*}
b\ln \left(\exp \left(\frac{h^0}{b} +  \ln (2)\right) - 1\right) + nd \times s \left(b\ln \left(\exp \left(\frac{\delta}{b} +  \ln (2)\right) - 1\right) \right) &<& h^0 + b \ln 2 + nd\sigma 
\end{eqnarray*}
and applying Equation \eqref{boundSigma} the condition for a project to be funded at equilibrium is:  
\begin{eqnarray*}
h^0 + b \ln 2 + nd\sigma < \vartheta_N
\end{eqnarray*}

\end{document}